\documentclass[10pt, oneside]{article}   
\usepackage{geometry}                	
\usepackage[parfill]{parskip}   
\usepackage{fixltx2e}	
\usepackage{hyperref}
\usepackage[T1]{fontenc}
\usepackage{amsmath}
\usepackage{textcomp}
\usepackage{textgreek}	
\usepackage{amssymb}
\usepackage{listings}
\usepackage[T1]{fontenc}
\usepackage[scaled=0.9]{DejaVuSansMono}
\usepackage{caption}
\usepackage{subcaption}
\usepackage{geometry}
\usepackage[T1]{fontenc}
\usepackage{titling}
\usepackage{comment} 
\usepackage{cleveref} 
\usepackage{tikz}
\usepackage{graphicx}
\usepackage{amsthm}
\usepackage{mathtools}
\usepackage{cancel}
\usepackage{bbm}

\usepackage{amsmath,amsfonts,bm}

\def\eqref#1{equation~\ref{#1}}

\def\1{\bm{1}}

\DeclareMathAlphabet{\mathsfit}{\encodingdefault}{\sfdefault}{m}{sl}
\SetMathAlphabet{\mathsfit}{bold}{\encodingdefault}{\sfdefault}{bx}{n}

\DeclareMathOperator*{\argmax}{arg\,max}

\DeclareFixedFont{\ttb}{T1}{txtt}{bx}{n}{12}
\DeclareFixedFont{\ttm}{T1}{txtt}{m}{n}{12}

\title{
\bfseries Notes on the runtime of A* sampling
\vspace*{-0.5cm}
}
\author{\begin{tabular}{c}
    \\
    Stratis Markou \\
    Department of Engineering \\
    University of Cambridge \\
    \texttt{em626@cam.ac.uk}
\end{tabular}}
\date{\today}

\newcommand{\Dinf}[2]{D_{\infty}[#1||#2]}
\newcommand{\rmax}{r_{max}}

\DeclarePairedDelimiterX{\innerProd}[2]{\langle}{\rangle}{%
    #1,#2%
}

\newtheorem{theorem}{Theorem}
\newtheorem{lemma}{Lemma}
\newtheorem{definition}{Definition}
\newtheorem{assumption}{Assumption}
\newtheorem{corollary}{Corollary}

\newcommand\eqdef{\stackrel{\tiny \text{def}}{=}}
\newcommand\eqdist{\stackrel{\tiny \text{d}}{=}}

\begin{document}
\maketitle

\begin{abstract}
    \noindent
    The challenge of simulating random variables is a central problem in Statistics and Machine Learning.
    Given a tractable proposal distribution $P$, from which we can draw exact samples, and a target distribution $Q$ which is absolutely continuous with respect to $P$, the A* sampling algorithm \cite{maddison2014sampling} allows simulating exact samples from $Q$, provided we can evaluate the Radon-Nikodym derivative of $Q$ with respect to $P$.
    Maddison et al. \cite{maddison2014sampling} originally showed that for a target distribution $Q$ and proposal distribution $P$, the runtime of A* sampling is upper bounded by $\mathcal{O}\left(\exp(\Dinf{Q}{P})\right)$ where $\Dinf{Q}{P}$ is the Renyi divergence from $Q$ to $P$.
    This runtime can be prohibitively large for many cases of practical interest.
    Here, we show that with additional restrictive assumptions on $Q$ and $P$, we can achieve much faster runtimes.
    Specifically, we show that if $Q$ and $P$ are distributions on $\mathbb{R}$ and their Radon-Nikodym derivative is unimodal, the runtime of A* sampling is $\mathcal{O}\left(\Dinf{Q}{P}\right)$, which is exponentially faster over A* sampling without further assumptions.
\end{abstract}

\section*{Introduction \& Related Work}

\textbf{Introduction:}
We give a proof for the linear runtime of A* sampling \cite{maddison2014sampling} on unimodal density ratios.
We will not specify how A* sampling works, but will define relevant notation.
For an exact specification of A* sampling, we refer the reader to the original paper \cite{maddison2014sampling}.

\textbf{Related work:}
The runtime of A* sampling with unimodal Radon-Nikodym derivatives has been studied in the context of A* coding, a compression algorithm based on A* sampling \cite{flamich2022fast}.
However, due to an inaccuracy in the analysis of the runtime of that algorithm, the proof provided in that work is not correct.
The present work provides a proof for this result.


\section*{Proof of linear runtime for A* sampling}


\begin{definition}[Proposal, Target]
\label{def:qp}
We define $Q$ and $P$ to be a target and a proposal probability measure defined on the Borel sigma algebra of $\mathbb{R}$ with $P \gg Q$.
\end{definition}


\begin{assumption}[Continuous distributions, finite $D_{\infty}$] \label{assumption:runtime_proof}
We assume that measures $Q$ and $P$ describe continuous random variables, so their densities $q$ and $p$ exist.
Since $P \gg Q$, the Radon-Nikodym derivative $r(x) = (dQ/dP)(x)$ also exists.
We also assume $r(x)$ is unimodal and satisfies
\begin{equation}
    \Dinf{Q}{P} = \log \sup_{x\in \mathbb{R}} \frac{dQ}{dP}(x) = \log r_{max} < \infty.
\end{equation}
\end{assumption}

Without loss of generality, we can also assume $P$ to be the uniform measure on $[0, 1]$, as shown by the next lemma.
This is because we can push $P$ and $Q$ through the CDF of $P$ to ensure $P$ is uniform, while leaving the Radon-Nikodym derivative unimodal and the $\infty$-divergence unchanged. \\

\begin{lemma}[Without loss of generality, $P$ is uniform]
Suppose $Q$ is a target measure and $P$ a proposal measure as specified in Assumption \ref{assumption:runtime_proof}.
Let $\Phi$ be the CDF associated with $P$ and consider the measures $P', Q' : [0, 1] \to [0, \infty)$ defined as
\begin{equation}
    P' = P \circ \Phi^{-1} ~\text{ and }~ Q' = Q \circ \Phi^{-1}.
\end{equation}
Then, $P'$ is the uniform measure on $[0, 1]$.
Further, the Radon-Nikodym derivative $dQ'/dP'(x)$ is unimodal, and
\begin{equation}
    \log \sup_{z \in [0, 1]} \frac{dQ'}{dP'}(z) = \log \sup_{x \in \mathbb{R}} \frac{dQ}{dP}(x).
\end{equation}
\end{lemma}
\begin{proof}
    First, $P'$ is the uniform measure on $[0, 1]$ since for any $z \in [0, 1]$
    \begin{equation}
        P'([0, z]) = P \circ \Phi^{-1}([0, z]) = P((-\infty, \Phi^{-1}(z)]) = [0, z].
    \end{equation}
    Now, let the densities of $Q$ and $P$ be $q$ and $p$, and the densities of $Q'$ and $P'$ be $q'$ and $p'$.
    Then by the change of variables formula
    \begin{equation}
        p'(z) = p\left(\Phi^{-1}(z)\right) (\Phi^{-1})'(z) ~~\text{ and }~~ q'(z) = q\left(\Phi^{-1}(z)\right) (\Phi^{-1})'(z).
    \end{equation}
    Therefore, we have
    \begin{equation}
        \frac{dQ'}{dP'} (z) = \frac{q'(z)}{p'(z)} = \frac{q\left(\Phi^{-1}(z)\right)}{p\left(\Phi^{-1}(z)\right)} = \frac{dQ}{dP} \circ \Phi^{-1}(z),
    \end{equation}
    Now, since $dQ/dP(x)$ is a unimodal function of $x$ and $\Phi^{-1}(z)$ is increasing in $z$, the function $(dQ/dP) \circ \Phi^{-1}(z)$ is unimodal in $z$.
    Also, by taking the the supremum and logarithm of both sizes
    \begin{equation}
        \log \sup_{z \in [0, 1]} \frac{dQ'}{dP'} (z) = \log \sup_{z \in [0, 1]} \frac{dQ}{dP} \circ \Phi^{-1}(z) = \log \sup_{x \in \mathbb{R}} \frac{dQ}{dP}(x),
    \end{equation}
    arriving at the result.
\end{proof}


\begin{definition}[Superlevel set, width]
We define the superlevel-set function $S : [0, 1] \to 2^{[0, 1]}$ as
\begin{equation}
    S(\gamma) = \{x \in [0, 1]~|~ r(x) \geq \gamma \rmax \},
\end{equation}
And let $x_{max} \in \{x \in [0, 1]~|~r(x) \leq r(x_{max})\}$ be an arbitrary maximiser of the density ratio.
We also define the superlevel-set width function $w : [0, 1] \to [0, 1]$ as
\begin{equation}
    w(\gamma) = \inf \{\delta \in [0, 1]~|~\exists z \in [0, 1],~ S(\gamma) \subseteq [z, z + \delta] \}.
\end{equation}
\end{definition}


\begin{lemma}[Properties of $w$] \label{lemma:propw}
The width function $w(\gamma)$ is non-increasing in $\gamma$ and satisfies
\begin{equation}
    \int_0^1 w(\gamma)~d\gamma = \frac{1}{\rmax} \text{ and } w(0) \geq \frac{1}{r_{max}}.
\end{equation}
\end{lemma}
\begin{proof}
    First, we note that if $\gamma_1 \leq \gamma_2$, then $S(\gamma_2) \subseteq S(\gamma_1)$ which implies $w(\gamma_2) \leq w(\gamma_1)$.
    Therefore
    \begin{equation}
         \gamma_1 \leq \gamma_2 \implies w(\gamma_2) \leq w(\gamma_1),
    \end{equation}
    so $w(\gamma)$ is decreasing in $\gamma$.
    Second, let $A = [0, 1] \times [0, \rmax]$, define $B = \left\{ (x, y) \in A ~|~ y \leq r(x) \right\}$ and
    consider the integral
    \begin{equation}
        I = \int_A \mathbbm{1}(z \in B)~dz.
    \end{equation}
    Since this the integrand is a non-negative measurable function, by Fubini's theorem, we have
    \begin{align}
        I = \int_0^1 \int_0^{\rmax} \mathbbm{1}((x, y) \in B)~dy~dx &= \int_0^{\rmax} \int_0^1 \mathbbm{1}((x, y) \in B)~dx~dy \\
        \int_0^1 r(x)~dx &= \int_0^{\rmax} w(y/\rmax)~dy \\
        \int_0^1 q(x)~dx &= \int_0^1 w(\gamma) \rmax~d\gamma \\
        \int_0^1 w(\gamma)~d\gamma &= \rmax^{-1}.
    \end{align}
    Last, since $w$ is non-increasing, we have $w(0) \geq \rmax^{-1}$, because otherwise $\int_0^1 w(\gamma)~ d\gamma < \rmax^{-1}$.
\end{proof}


\begin{definition}[\# steps to $S(\gamma)$, \# residual steps]
Suppose AS* is applied to a target-proposal pair $Q, P$ satisfying Assumption \ref{assumption:runtime_proof}, producing a sequence of samples $X_1, X_2, \dots$.
We use $T \in \mathbb{Z}$ to denote the total number of steps taken by AS* until it terminates and define the random variables
\begin{equation}
    N(\gamma) = \min \{n \in \mathbb{Z}~|~X_n \in S(\gamma)\} ~\text{ and }~ K(\gamma) = \max\{0, T - N(\gamma)\}.
\end{equation}
\end{definition}


\begin{lemma}[A* sampling with a unimodal density ratio follows nested bounds]
    Let $X_1, X_2, \dots$ be the sequence of samples produced by A* sampling, and let $x^* = \argmax_{x \in \mathbb{R}} r(x)$ be a mode of the Radon-Nikodym derivative of $Q$ with respect to $P$.
    Also for $n = 1, 2, \dots$ let
    \begin{equation}
        B^L_{n+1} = \max\{X_k~|~k \leq n, X_k \leq x^*\} ~\text{ and }~ B^R_{n+1} = \min\{X_k~|~k \leq n, X_k \geq x^*\}
    \end{equation}
    and $B_0^L = 0, B_0^R = 1$.
    Also let $B_n = [B^L_n, B^R_n]$.
    Then $B_1, B_2, \dots$ is a sequence of nested intervals, that is $B_1 \supseteq B_2 \supseteq \dots$.
\end{lemma}
\begin{proof}
    This is shown in the original A* sampling paper \cite{maddison2014sampling}.
\end{proof}


\begin{lemma}[Bounds on the expected $P(B_n)$]\label{lemma:3/4}
    Let $B_1, B_2, \dots$ be the bounds produced by A*, and define $Z_n = P(B_n)$.
    Then
    \begin{equation}
        \left(\frac{1}{2}\right)^{n-1} \leq ~\mathbb{E}[Z_n] ~\leq \left(\frac{3}{4}\right)^{n-1}, \text{ for } n = 1, 2, \dots .
    \end{equation}
\end{lemma}
\begin{proof}
    For a proof of this see \cite{flamich2022fast}.
\end{proof}


\begin{lemma}[Upper bound on the probability of $P(B_n) \geq w(\gamma)$] \label{lemma:zw}
    Let $Z_n = P(B_n)$.
    Then
    \begin{equation}
        \mathbb{P}(Z_n \geq w(\gamma)) \leq \frac{1}{w(\gamma)}\left(\frac{3}{4}\right)^{n-1}
    \end{equation}
\end{lemma}
\begin{proof}
    Let $Z_n = P(B_n)$.
    Noting that $Z_n \geq 0$, we apply Markov's inequality and \cref{lemma:3/4} to get
    \begin{equation}
        \mathbb{P}(Z_n \geq w(\gamma)) \leq \frac{1}{w(\gamma)}\mathbb{E}[Z_n] \leq \frac{1}{w(\gamma)}\left(\frac{3}{4}\right)^{n-1},
    \end{equation}
    as required.
\end{proof}


\begin{lemma}[Bound on expected $N(\gamma)$] \label{lemma:n_bound}
The random variable $N(\gamma)$ satisfies
\begin{equation}
    \mathbb{E}[N(\gamma)] \leq \alpha \log \frac{1}{w(\gamma)} + 6, \text{ where } \alpha =  \left(\log\frac{4}{3} \right)^{-1}.
\end{equation}
\end{lemma}
\begin{proof}
    Let $N_0 = \left\lceil \frac{\log w(\gamma)}{\log (3/4)} \right\rceil$ + 1.
    Also let $B_0, B_1, \dots$ be the bounds produced by AS*.
    Noting that by the unimodality of $r$, $S(\gamma)$ is an interval with $x_{max} \in S(\gamma)$, and $x_{max} \in B_n$, we have
    \begin{equation}
        P(B_n) < w(\gamma) \implies N(\gamma) \leq n,
    \end{equation}
    that is, the event $P(B_n) < w(\gamma)$ implies the event $N(\gamma) \leq n$.
    From this it follows that
    \begin{equation}
        \mathbb{P}(P(B_n) < w(\gamma)) \leq \mathbb{P}(N(\gamma) \leq n) \implies \mathbb{P}(P(B_n) \geq w(\gamma)) \geq \mathbb{P}(N(\gamma) \geq n + 1). \label{eq:b_less_than_w}
    \end{equation}
    Using this together with \cref{lemma:zw}, we can write
    \begin{align}
        \mathbb{E}_{B_{0:\infty}}[N(\gamma)] &= 
        \sum_{n=1}^\infty \mathbb{P}\left(N(\gamma)=n\right)~n \label{eq:ineq:n1} \\
        &= \sum_{n=1}^\infty \mathbb{P}\left(N(\gamma) \geq n\right) \label{eq:ineq:n2} \\
        &\leq N_0 + \sum_{n=N_0 + 1}^\infty \mathbb{P}\left(N(\gamma) \geq n\right) \label{eq:ineq:n3} \\
        &= N_0 + \sum_{n=1}^\infty \mathbb{P}\left(N(\gamma) \geq N_0 + n\right) \label{eq:ineq:n4} \\
        &\leq N_0 + \sum_{n=1}^\infty \mathbb{P}(B_{N_0+n-1} \geq w(\gamma)) \label{eq:ineq:n5} \\
        &\leq N_0 + \sum_{n=1}^\infty \frac{1}{w(\gamma)}\left(\frac{3}{4}\right)^{N_0+n-1} \label{eq:ineq:n6} \\
        &\leq N_0 + \sum_{n=1}^\infty \left(\frac{3}{4}\right)^n \label{eq:ineq:n7} \\
        &= N_0 + 4 \label{eq:ineq:n8} \\
        &\leq \frac{\log w(\gamma)}{\log 3/4} + 6,
    \end{align}
    where the equality of \ref{eq:ineq:n1} and \ref{eq:ineq:n2} is a standard identity \cite{grimmett2014probability}, \ref{eq:ineq:n2} to \ref{eq:ineq:n3} follows by the fact that probabilities are smaller than 1, \ref{eq:ineq:n3} to \ref{eq:ineq:n4} follows by relabelling the indices, \ref{eq:ineq:n4} to \ref{eq:ineq:n5} follows from \cref{eq:b_less_than_w}, \ref{eq:ineq:n5} to \ref{eq:ineq:n6} follows from follows from \cref{lemma:zw} and \ref{eq:ineq:n5} to \ref{eq:ineq:n6} follows from our definition of $N_0$.
\end{proof}


\begin{lemma}[Exponentials and Truncated Gumbels] \label{lemma:expg}
Let $T \sim \text{Exp}(\lambda)$ and $T_0 \geq 0$.
Then
\begin{equation}
    Z \eqdef - \log (T + T_0) \eqdist G ~\text{ where }~ G \sim \text{TG}(\log \lambda,  -  \log T_0).
\end{equation}
\end{lemma}
\begin{proof}
    Let $T \sim \text{Exp}(\lambda)$, $T_0 \geq 0$ and define
    \begin{equation}
        Z \eqdef - \log (T + T_0).
    \end{equation}
    We note that $Z \leq - \log T_0$.
    For $Z \leq - \log T_0$, we can apply the change of variables formula to obtain the density of $Z$.
    Let $p_Z$ and $p_T$ be the densities of $Z$ and $T$. 
    Then
    \begin{align}
        p_Z(z) &= p_T(t) \left| \frac{dt}{dz}\large \right| \\
               &= \lambda e^{-\lambda t} \left| \frac{d}{dz} (e^{-z} - T_0) \right| \\
               &\propto e^{-\lambda e^{-z}} e^{-z} \\
               &\propto e^{-z - e^{-(z - \kappa)}},
    \end{align}
    where $\kappa = \log \lambda$.
    Therefore $Z$ has distribution $\text{TG}(\log \lambda, - \log T_0)$.
\end{proof}


\begin{lemma}[Mean of exponentiated negative truncated Gumbel] \label{lemma:mean_neg_g}
Let $B_1, \dots, B_N$ be the first $N$ bounds produced by AS*, let $G_N$ be the $N^{th}$ Gumbel produced by AS* and define $E_N = e^{-G_N}$.
Then
\begin{equation}
    \mathbb{E}[E_N~|~B_{1:N}] = \sum_{n=1}^N\frac{1}{P(B_n)}.
\end{equation}
\end{lemma}
\begin{proof}
    Define $E_n = e^{-G_n}$ for $n=1, 2, \dots$.
    By the definition of AS*, we have
    \begin{equation}
        G_N ~|~ G_{N-1}, B_{1:N} \sim \text{TG}(\log P(B_N), G_{N-1}).
    \end{equation}
    Negating $G_N$ and $G_{N-1}$, taking exponentials and applying Jensen's inequality together with ineq. (\ref{lemma:expg}), we obtain
    \begin{equation}
        E_N ~|~ T_{N-1}, B_{0:N} \eqdist \tau_N + T_{N-1}, \text{ where } \tau_{N-1} \sim \text{Exp}(P(B_N)).
    \end{equation}
    Repeating this step and taking expectations, we have
    \begin{equation}
        \mathbb{E}[E_N ~|~ B_{0:N}] = \mathbb{E}\left[\sum_{n=1}^N \tau_n ~\big|~ B_{0:N}\right] = \sum_{n=1}^N \frac{1}{P(B_n)}
    \end{equation}
    as required.
\end{proof}


\begin{lemma}[Bound on expected $K(\gamma)$] \label{lemma:k_bound}
    The random variable $K(\gamma)$ satisfies
    \begin{equation}
        \mathbb{E}[K(\gamma)] \leq \alpha\left(\log \frac{1}{\gamma} + \log \frac{1}{w(\gamma)}\right) + 16, \text{ where } \alpha = \left(\log \frac{4}{3}\right)^{-1}.
    \end{equation}
\end{lemma}
\begin{proof}
    Let the global upper and lower bounds of AS* at step $n$ be $U_n$ and $L_n$ respectively.
    Then, by the definition of the upper bound of AS* coding
    \begin{equation}
        U_{N(\gamma)} = \log r_{max} + G_{N(\gamma)},
    \end{equation}
    and also, by the definition of the lower bound of AS* coding
    \begin{equation}
        L_{N(\gamma)} = \max_{n \in [1:N(\gamma)]}\big\{\log r(x_n) + G_n \big\} \geq \log r\left(x_{N(\gamma)}\right) + G_{N(\gamma)} \geq \log \gamma r_{max} + G_{N(\gamma)}.
    \end{equation}
    Now for $k = 0, 1, 2, \dots$, we have
    \begin{equation}
        U_{N(\gamma) + k} - L_{N(\gamma) + k} \leq 0 ~\implies~ T \leq N(\gamma) + k,
    \end{equation}
    that is, the event $U_{N(\gamma) + k} - L_{N(\gamma) + k} \leq 0$ implies the event $T \leq N(\gamma) + k$.
    This is because if $U_{N(\gamma) + k} - L_{N(\gamma) + k} \leq 0$, then the algorithm has terminated by step $N(\gamma) + k$, so it follows that $T \leq N(\gamma) + k$.
    Further
    \begin{align}
        U_{N(\gamma) + k} - L_{N(\gamma) + k} &\leq U_{N(\gamma) + k} - L_{N(\gamma)} \\
        &\leq \log r_{max} + G_{N(\gamma)+k} - \log \gamma r_{max} - G_{N(\gamma)} \\
        &= \log \frac{1}{\gamma} + G_{N(\gamma)+k} -  G_{N(\gamma)}.
    \end{align}
    Therefore, we have
    \begin{equation}
        G_{N(\gamma)+k} -  G_{N(\gamma)} \leq \log \gamma \implies T \leq N(\gamma) + k \implies K(\gamma) \leq k,
    \end{equation}
    that is, the event $G_{N(\gamma)+k} -  G_{N(\gamma)} \leq \log \gamma$ implies the event $K(\gamma) \leq k$.
    This holds because if $G_{N(\gamma)+k} -  G_{N(\gamma)} \leq \log \gamma$, then $U_{N(\gamma) + k} - L_{N(\gamma) + k} \leq 0$, which in turn implies $K(\gamma) \leq k$.
    Therefore
    \begin{equation} \label{eq:gk}
        \mathbb{P}\left(G_{N(\gamma)+k} -  G_{N(\gamma)} \leq \log \gamma \right) \leq \mathbb{P}\left(K(\gamma) \leq k \right) \implies
        \mathbb{P}\left(G_{N(\gamma)+k} -  G_{N(\gamma)} \geq \log \gamma \right) \geq \mathbb{P}\left(K(\gamma) \geq k+1\right).
    \end{equation}
    \Cref{eq:gk} upper bounds the probability that the second stage of the algorithm has not terminated, by the probability that the Gumbel values have decreased sufficiently.
    To proceed, we turn to lower bounding the probability of the complementary event $G_{N(\gamma)+k} -  G_{N(\gamma)} \leq \log \gamma$.
    Let $\Phi_{TG}(g; \mu, \kappa)$ denote the CDF of a truncated Gumbel distribution with location parameter $\mu$ and unit scale parameter, truncated at $\kappa$.
    Then
    {
    }
    \begin{align}
        \mathbb{P}\big( &G_{N(\gamma)+n} - G_{N(\gamma)} \leq \log \gamma ~|~ N(\gamma), G_{N(\gamma)}, B_{0:N(\gamma)+n} \big) = \\
        &= \mathbb{E}_{G_{N(\gamma)+n-1}}\left[\mathbb{P}\big( G_{N(\gamma)+n} - G_{N(\gamma)} \leq \log \gamma ~|~N(\gamma), G_{N(\gamma)}, G_{N(\gamma)+n-1}, B_{0:N(\gamma)+n} \big)\right] \\
        &= \mathbb{E}_{G_{N(\gamma)+n-1}}\left[\Phi_{TG}\left(\log \gamma + G_{N(\gamma)};~ \log P(B_{N(\gamma)+n}),~ G_{N(\gamma)+n-1} \right)\right] \\
        &\geq \Phi_{TG}\left(\log \gamma + G_{N(\gamma)};~ \log P(B_{N(\gamma)+n}),~ \infty \right) \\
        &= e^{-e^{- \left(\log \gamma + G_{N(\gamma)} - \log P\left(B_{N(\gamma)+n}\right)\right)}} \\
        &= e^{-\frac{1}{\gamma}~ P\left(B_{N(\gamma)+n}\right)~ e^{-G_{N(\gamma)}}}.
    \end{align}
    Taking an expectation over $G_{N(\gamma)}$ and $ B_{0:N(\gamma)+n}$, we have
    \begin{align}
        \mathbb{P}\big( G_{N(\gamma)+n} - G_{N(\gamma)} \leq \log \gamma ~|~ N(\gamma) \big) &\geq \mathbb{E}_{G_{N(\gamma)}, B_{0:N(\gamma)+n}}\left[ e^{-\frac{1}{\gamma}~ P\left(B_{N(\gamma)+n}\right)~ e^{-G_{N(\gamma)}}} ~\Big|~ N(\gamma) \right] \\
        &\geq e^{-\frac{1}{\gamma}~ \mathbb{E}_{G_{N(\gamma)}, B_{0:N(\gamma)+n}}\left[P\left(B_{N(\gamma)+n}\right)~ e^{-G_{N(\gamma)}}~|~ N(\gamma) \right]} \label{eq:K1}
    \end{align}
    
    Focusing on the term in the exponent, we have
    \begin{align}
        &\mathbb{E}_{G_{N(\gamma)}, B_{0:N(\gamma)+n}}\left[P\left(B_{N(\gamma)+n}\right)~ e^{-G_{N(\gamma)}}~\big|~ N(\gamma) \right] = \\
        &= \mathbb{E}_{B_{0:N(\gamma)-1}}\left[\mathbb{E}_{G_{N(\gamma)}, B_{N(\gamma):N(\gamma)+n}}\left[ P\left(B_{N(\gamma)+n}\right)~ e^{-G_{N(\gamma)}} ~\big|~ B_{0:N(\gamma)-1}, N(\gamma) \right]~\Big|~ N(\gamma) ~\right] \\
        &\leq \mathbb{E}_{B_{0:N(\gamma)-1}}\left[\mathbb{E}_{G_{N(\gamma)}}\left[ \left(\frac{3}{4}\right)^{n+1} P\left(B_{N(\gamma)-1}\right)~ e^{-G_{N(\gamma)}} ~\bigg|~ B_{0:N(\gamma)-1} \right]~\bigg|~ N(\gamma) ~\right] \\
        &= \mathbb{E}_{B_{0:N(\gamma)-1}}\left[ \left(\frac{3}{4}\right)^{n+1} P\left(B_{N(\gamma)-1}\right) \sum_{n=0}^{N(\gamma)-1} \frac{1}{P(B_n)} ~\bigg|~ N(\gamma) ~\right] \\
        &\leq \left(\frac{3}{4}\right)^{n+1} N(\gamma). \label{eq:K2}
    \end{align}
    Substituting \cref{eq:K2} into \cref{eq:K1}, we obtain
    \begin{equation}
        \mathbb{P}\big( G_{N(\gamma)+n} - G_{N(\gamma)} \leq \log \gamma ~|~ N(\gamma) \big) \geq e^{-\frac{N(\gamma)}{\gamma}~ \left(\frac{3}{4}\right)^{n+1}}
    \end{equation}
    and applying ineq. (\ref{eq:ineq:n8}) to this we obtain
    \begin{equation} \label{eq:ineq:gkn}
        \mathbb{P}\big( G_{N(\gamma)+n} - G_{N(\gamma)} \leq \log \gamma \big) \geq e^{-\frac{1}{\gamma}~ \left(\frac{3}{4}\right)^{n+1}~\left(N_0 + 4 \right)},
    \end{equation}
    arriving at a lower bound on which does not depend on any random quantities.
    Now we also have
    \begin{align}
        \log \frac{1}{\gamma} + \log(N_0 + 4) &= \log \frac{1}{\gamma} + \log\left(\left\lceil \frac{\log w(\gamma)}{\log (3/4)} \right\rceil + 5\right) \\
        &\leq \log \frac{1}{\gamma} + \log\left( \frac{\log w(\gamma)}{\log (3/4)} + 6\right) \label{eq:ineq:loglog} \\
        &\leq \log \frac{1}{\gamma} + \log \frac{1}{w(\gamma)} + 2, \label{eq:ineq:singlelog}
    \end{align}
    where going from \ref{eq:ineq:loglog} to \ref{eq:ineq:singlelog} can be verified numerically.
    Therefore, letting $K_0 = \left\lceil \frac{\log (1/\gamma)~+~\log (1/w(\gamma))~+~2}{\log (4/3)} \right\rceil$
    \begin{align}
        \mathbb{E}\left[K(\gamma)\right] &= \sum_{k=0}^\infty \mathbb{P}\big( K(\gamma) = k \big) ~k \label{eq:ineq:k1} \\
        &= \sum_{k=0}^\infty \mathbb{P}\big( K(\gamma) \geq k \big)  \label{eq:ineq:k2} \\
        &\leq K_0  + \sum_{k=K_0 + 1}^\infty \mathbb{P}\big( K(\gamma) \geq  k \big)  \label{eq:ineq:k3} \\
        &= K_0 + \sum_{k=1}^\infty \mathbb{P}\big( K(\gamma) \geq K_0 + k \big) \label{eq:ineq:k4} \\
        &\leq K_0 + \sum_{k=1}^\infty \mathbb{P}\big( G_{N(\gamma)+K_0+k-1} - G_{N(\gamma)} > \log \gamma \big) \label{eq:ineq:k5} \\
        &\leq K_0 + \sum_{k=1}^\infty \left(1 - e^{-\left(\frac{3}{4}\right)^k}\right) \label{eq:ineq:k6} \\
        &\leq K_0 + 4 \label{eq:ineq:k7} \\
        &\leq \frac{\log \gamma}{\log (3/4)} + \frac{\log w(\gamma)}{\log (3/4)} + 16. \label{eq:ineq:k8}
    \end{align}
    where the equality of \ref{eq:ineq:k1} and \ref{eq:ineq:k2} is a standard identity \cite{grimmett2014probability}, \ref{eq:ineq:k2} to \ref{eq:ineq:k3} follows because probabilities are bounded above by $1$, \ref{eq:ineq:k3} to \ref{eq:ineq:k4} follows by relabelling the indices, \ref{eq:ineq:k4} to \ref{eq:ineq:k5} follows by ineq. (\ref{eq:gk}), \ref{eq:ineq:k5} to \ref{eq:ineq:k6} follows by ineq. (\ref{eq:ineq:gkn}) and the definition of $K_0$, \ref{eq:ineq:k6} to \ref{eq:ineq:k7} can be verified by evaluating the sum using numerical means and \ref{eq:ineq:k7} to \ref{eq:ineq:k8} follows by the definition of $K_0$.
\end{proof}


\begin{corollary}[Upper bound on $T$ for given $w$] \label{corollary:t_bound}
    For any $\gamma \in [0, 1]$, the total number of steps, $T$, satisfies
    \begin{equation}
        \mathbb{E}[T] \leq 2\alpha \left(\log \frac{1}{w(\gamma)} + 2 \log \frac{1}{\gamma} \right)+ 22, ~\text{ where }~ \alpha = \left(\log \frac{4}{3}\right)^{-1}
    \end{equation}
\end{corollary}
\begin{proof}
    By the definition of $N(\gamma)$ and $K(\gamma)$, we have
    \begin{equation}
        T \leq N(\gamma) + K(\gamma) \implies \mathbb{E}[T] \leq \mathbb{E}\left[N(\gamma) + K(\gamma)\right],
    \end{equation}
    for all $\gamma \in [0, 1]$.
    From \cref{lemma:n_bound} and \cref{lemma:k_bound}, we have
    \begin{equation} \label{eq:looser}
        \mathbb{E}[T] \leq \alpha \left(2\log \frac{1}{w(\gamma)} + \log \frac{1}{\gamma} \right)+ 22 \leq 2\alpha \left(\log \frac{1}{w(\gamma)} + 2\log \frac{1}{\gamma} \right)+ 22,
    \end{equation}
    where $\alpha = \log (4/3)^{-1}$ as required.
\end{proof}


\begin{definition}[Bound functions $f, g, h$, worst-case width set $W^*$] \label{def:fghw}
    We define
    \begin{equation}
        f(\gamma, w) = \log \frac{1}{w(\gamma)},~ g(\gamma) = 2\log \frac{1}{\gamma} ~ \text{ and } ~ h(\gamma, w) = f(\gamma, w) + g(\gamma).
    \end{equation}
    For fixed $\rmax$, let $W(\rmax)$ be the set of all possible width functions.
    We define the set $W^*$ of worst-case width functions as
    \begin{equation}
        W^* = \left\{w^* \in W(\rmax) ~\Big|~ \inf_{\gamma'} h(\gamma', w^*) \geq \inf_{\gamma'} h(\gamma', w) ~\forall~ w \in W(\rmax) \right\}.
    \end{equation}
    We refer to members of this set as worst-case width functions.
\end{definition}


\begin{lemma}[An explicit worst case width function] \label{lemma:wcw}
    The function
    \begin{equation} \label{eq:w_form}
        \tilde{w}(\gamma) = \begin{cases}
            1 & \text{ for } 0 \leq \gamma \leq \tilde{\gamma} \\
            (\tilde{\gamma}/\gamma)^2 & \text{ for } \tilde{\gamma} < \gamma \leq 1
        \end{cases},
    \end{equation}
    where $\tilde{\gamma} = 1 - \sqrt{1 - \rmax^{-1}}$,
    is a width function and $\tilde{w} \in W(\rmax)$.
    Further, if $w \in W(\rmax)$ then
    \begin{equation}
        \inf_{\gamma} h(\gamma, \tilde{w}) \geq \inf_{\gamma} h(\gamma, w).
    \end{equation}
\end{lemma}
\begin{proof}
    Suppose $w \in W(\rmax)$ and let
    \begin{equation}
        m = \inf_{\gamma} h(\gamma, \tilde{w}),
    \end{equation}
    let $\gamma_m$ be the point where $g$ equals $m$, that is
    \begin{equation}
        g(\gamma_m) = 2\log \frac{1}{\gamma_m} = m \implies \gamma_m = e^{-m/2}.
    \end{equation}
    Define $v : [0, 1] \times [0, 1] \to [0, 1]$ as
    \begin{equation}
        v(\gamma, \gamma') = \begin{cases}
            1 & 0 \leq \gamma \leq \gamma' \\
            (\gamma'/\gamma)^2 & \gamma' < \gamma \leq 1
        \end{cases},
    \end{equation}
    and consider $v(\gamma, \gamma_m)$ as a function of $\gamma$.
    Note that $v(\gamma, \gamma_m)$ may not be in $W(\rmax)$ because, while it is non-increasing and continuous, it may not integrate to $\rmax^{-1}$.
    In particular it holds that
    \begin{equation}
        h(\gamma, v(\gamma, \gamma_m)) \leq h(\gamma, \tilde{w}) ~\text{ for all }~ \gamma \in [0, 1] \implies
        v(\gamma, \gamma_m) \geq w(\gamma) \implies \int_0^1 v(\gamma, \gamma_m)~d\gamma \geq \rmax^{-1}.
    \end{equation}
    Now, note that
    \begin{equation}
        \int_0^1 v(\gamma, \gamma')~d\gamma = 2\gamma' - (\gamma')^2.
    \end{equation}
    By the intermediate value theorem, there exists some $0 < \tilde{\gamma} \leq \gamma_m$ such that $2\tilde{\gamma} - \tilde{\gamma}^2 = \rmax^{-1}$.
    For this $\tilde{\gamma}$, we define $\tilde{w}(\gamma) = v(\gamma, \tilde{\gamma})$, which is a width function because it is decreasing and integrates to $1$.
    Specifically, $\tilde{w}(\gamma)$ is in $W(\rmax)$ because the probability density function
    \begin{equation}
        q(x) = \rmax \min\left\{1, \tilde{\gamma} x^{-1/2}\right\},
    \end{equation}
    has $\tilde{w}(\gamma)$ as its width function.
    In addition note that $\tilde{\gamma} \leq \gamma_m$ so we have
    \begin{equation}
        \tilde{\gamma} \leq \gamma_m \implies \inf_\gamma h(\gamma, \tilde{w}(\gamma)) = \inf_\gamma h(\gamma, v(\gamma, \tilde{\gamma})) \geq \inf_\gamma h(\gamma, v(\gamma, \gamma_m)) = \inf_\gamma h(\gamma, w(\gamma)).
    \end{equation}
    Therefore it holds that
    \begin{equation} \label{eq:wtilde}
        w \in W(\rmax) \implies \inf_\gamma h(\gamma, \tilde{w}) \geq \inf_\gamma h(\gamma, w),
    \end{equation}
    from which it follows that $\tilde{w} \in W^*$ is a width function.
\end{proof}


\begin{theorem}[AS* runtime upper bound]
    Let $T$ be the total number of steps taken by AS* until it terminates.
    Then
    \begin{equation}
        \mathbb{E}[T] \leq 2\alpha \log \rmax + 2\alpha \log 2 + 22.
    \end{equation}
\end{theorem}
\begin{proof}
    Suppose AS* is applied to a target $Q$ and proposal $P$ with $\Dinf{Q}{P} = \rmax$, and corresponding width function $w \in W(\rmax)$.
    Now consider the worst case width function $\tilde{w}$ defined in \cref{lemma:wcw}, and note
    \begin{equation}
        \tilde{\gamma} = 1 - \sqrt{1 - \rmax^{-1}} \geq \frac{1}{2\rmax}.
    \end{equation}
    Then we have
    \begin{align}
        \mathbb{E}[T] &\leq 2\alpha \inf_{\gamma} h(\gamma, w) + 22 \leq 2\alpha \inf_{\gamma} h(\gamma, \tilde{w}) + 22 \leq 2\alpha h(\tilde{\gamma}, \tilde{w}) + 22,
    \end{align}
    and substituting the expression for $h$ we obtain
    \begin{align}
        \mathbb{E}[T] &\leq 2\alpha \left(\log \frac{1}{\tilde{w}(\tilde{\gamma})} + 2\log \frac{1}{\tilde{\gamma}} \right) + 22 \leq 4\alpha \log \rmax + 4\alpha \log 2 + 22,
    \end{align}
    arriving at the result.
\end{proof}

\bibliographystyle{unsrt}
\bibliography{references}

\end{document}